\documentclass[journal]{IEEEtran}
\IEEEoverridecommandlockouts

\usepackage{cite}
\usepackage{xr}
\usepackage{xcite}
\usepackage{amsmath,amssymb,amsfonts}
\usepackage{multirow}
\usepackage{subfigure}
\usepackage{graphicx}
\usepackage{textcomp}
\usepackage{xcolor}
\usepackage{hyperref}
\usepackage{amsbsy}
\usepackage{amsthm}
\usepackage{algorithm,algpseudocode}
\usepackage{empheq}

\algnewcommand{\Inputs}[1]{%
	\State \textbf{Inputs:}
	\Statex \hspace*{\algorithmicindent}\parbox[t]{.8\linewidth}{\raggedright #1}
}
\algnewcommand{\Initialize}[1]{%
	\State \textbf{Initialize:}
	\Statex \hspace*{\algorithmicindent}\parbox[t]{.8\linewidth}{\raggedright #1}
}

\usepackage{bm}

\newcommand{\br}[1]{\left({#1}\right)}
\newcommand{\bc}[1]{\left\{{#1}\right\}}
\newcommand{\bsq}[1]{\left[{#1}\right]}
\renewcommand{\b}[1]{\ensuremath{\mathbf{#1}}} 
\newcommand{\norm}[1]{\ensuremath{\left\|#1\right\|}} 

\def \x {{\b{x}}}

\def \r {{\b{r}}}
\def \s {{\b{s}}}
\def \v {{\b{v}}}

\def \u {{\b{u}}}

\def \lam {{{\lambda}}}

\def \cL {{\mathcal{L}}}
\def \cN {{\mathcal{N}}}

\def \cQ {{\mathcal{Q}}}

\def \uu {{\b{U}}}
\def \uv {{\b{V}}}
\def \ux {{\b{X}}}
\def \uy {{\b{Y}}}
\def \uz {{\b{Z}}}

\def \Cn {{\mathbb{C}}}
\def \Gn {{\mathbb{G}}}
\def \In {{\mathbb{I}}}

\def \Rn {{\mathbb{R}}}

\theoremstyle{remark}
\newtheorem{rem}{\bf Remark}
\newtheorem{theorem}{Theorem}
\newtheorem{lemma}{Lemma}

\def\BibTeX{{\rm B\kern-.05em{\sc i\kern-.025em b}\kern-.08em
    T\kern-.1667em\lower.7ex\hbox{E}\kern-.125emX}}

\begin{document}

\title{Communication-Constrained Energy-Optimal Trajectory Generation for Quadrotor UAVs in Urban Environments}

\author{
    Prateek Priyaranjan Pradhan,
    Ketan Rajawat,
    and Mangal Kothari
    \thanks{
    Prateek Priyaranjan Pradhan is with the Department of Aerospace Engineering, Indian Institute of Technology Kanpur, Kanpur, India 
    (e-mail: prateekpp20@iitk.ac.in).
    }
    \thanks{
    Ketan Rajawat is with the Department of Electrical Engineering, Indian Institute of Technology Kanpur, Kanpur, India 
    (e-mail: ketan@iitk.ac.in).
    }
    \thanks{
    Mangal Kothari was with the Department of Aerospace Engineering, Indian Institute of Technology Kanpur, Kanpur, India. He is currently with ADASI, EDGE Group, Abu Dhabi, UAE. This work was carried out while he was with the Indian Institute of Technology Kanpur 
    (e-mail: mangalgnc@gmail.com).
    }
}

\maketitle

\begin{abstract}

Communication-aware trajectory generation for unmanned aerial vehicles (UAVs) operating in urban environments requires simultaneous consideration of vehicle dynamics, wireless communication quality, obstacle avoidance, and onboard energy limitations. In such missions, UAVs must navigate through obstacle-rich environments while ensuring reliable relay of mission-critical sensory information to ground infrastructure. This results in a highly nonlinear and nonconvex optimal control problem involving coupled communication and flight-dynamics constraints. This paper presents a communication-constrained energy-optimal trajectory generation framework for quadrotor UAVs operating in urban environments. The proposed formulation incorporates full rigid-body quadrotor dynamics, realistic urban wireless communication models, cumulative data throughput constraints, and obstacle avoidance requirements within a unified free-final-time optimal control framework. Unlike conventional approaches based on simplified kinematic or point-mass models, the proposed formulation generates dynamically feasible trajectories that remain directly implementable on practical aerial platforms.
The resulting nonconvex optimal control problem is solved iteratively using sequential convex programming (SCP), enabling computationally tractable trajectory generation while preserving dynamic feasibility and communication constraints. Numerical simulations are presented for multiple urban mission scenarios involving varying communication-throughput requirements and obstacle configurations. The results demonstrate the ability of the proposed framework to generate energy-efficient and communication-aware trajectories while adapting mission duration according to data relay requirements. The proposed methodology provides a practical framework for autonomous UAV operations requiring reliable communication in dense urban environments.

\end{abstract}


\section{Introduction}

Unmanned aerial vehicles (UAVs), also commonly known as drones, promise a wide variety of exciting applications in the modern world~\cite{survey}. In earlier decades, UAVs were mainly used in military applications and deployed in hostile territory for remote surveillance and armed attacks to reduce pilot losses~\cite{survey}. Thanks to recent advancements in digital communication, sensing technologies, and robotics, UAVs have extended their applications to civilian and commercial sectors~\cite{survey2}. With that, the demand for UAVs in aerial inspection, photography, search and rescue, package delivery, intelligent transportation monitoring, and persistent aerial surveillance is growing rapidly in urban cities.

Trajectory generation and autonomous navigation for UAVs operating in constrained urban environments have remained active areas of research over the past decade. Earlier works investigated robust UAV path planning in uncertain urban environments using probabilistically robust rapidly-exploring random trees (RRTs)~\cite{kothari2013rrt}, urban path-following strategies under wind disturbances~\cite{kothari2014urban}, and distributed probabilistic motion planning algorithms for multi-agent systems~\cite{kothari2013distributed}. In parallel, significant research efforts have focused on cooperative UAV systems, target-centric formations, distributed optimization, and multi-agent coordination strategies~\cite{sen2019circum,sen2022distributed,pradhan2023distributed,pradhan2024circumnavigation}. More recently, communication-aware autonomy and low-bandwidth UAV networking architectures have gained increasing attention due to their relevance in swarm robotics, distributed sensing, and network-enabled autonomous systems~\cite{samshad2025cnet}.

Due to the infrastructural developments in urban cities, generating a feasible and safe trajectory autonomously for a UAV is not only desirable but also necessary for the execution of any task. Failing to generate a safe trajectory that obeys the vehicular dynamics may result in loss of the vehicle payload or even harm to human lives in the environment. Along with generating safe trajectories, reliable UAV communication is essential in various applications such as urban traffic monitoring, infrastructure inspection, package delivery, and city surveillance. Existing literature broadly classifies UAV-enabled communication into two categories: UAV as an aerial base station~\cite{UAVasBS} providing seamless wireless coverage within the serving area, and UAV as a mobile relay or data collector~\cite{UAVasRelay}. The presented work focuses on the latter case in an urban environment. Such scenarios are common when the UAV has collected sensory data but processing the data onboard is inefficient due to limited computational resources. In these situations, the UAV can transmit the data to a ground station (GS) for further processing, saving energy for longer endurance. Figure~\ref{fig:urban_env_scene} illustrates a scenario similar to the problem discussed in this paper, where a UAV is required to generate a feasible and optimal trajectory from a specified starting point to a goal position in an urban environment while avoiding obstacles (buildings and trees) and relaying data to a GS.

\begin{figure}[H]
 \centerline{
 \includegraphics[width=0.6\columnwidth]{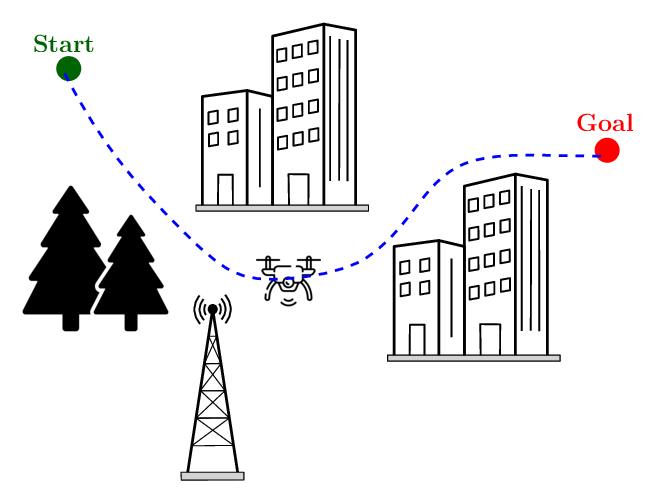}}
 \caption{{Illustration of a UAV mission scenario in an urban environment}}
 \label{fig:urban_env_scene}
\end{figure}

One critical issue in UAV-enabled wireless communication is the limited onboard energy of UAVs, which needs to be efficiently utilized to enhance communication performance and prolong the UAV's endurance. The authors of~\cite{zhang1st} were the first to propose an energy-efficient UAV communication model for a quadrotor UAV, where the energy consumption is derived as a function of UAV velocity. Adopting the energy consumption model in~\cite{zhang1st}, a robust resource allocation algorithm was investigated in~\cite{multiUser} by jointly optimizing the UAV trajectory and the transmit beamforming vector. Moreover, a 3D energy consumption model was proposed in~\cite{Sunetal}, assuming that the UAV moves smoothly with small acceleration and constant cruising speed.

However, these works consider the UAV as a point mass model and do not take into account the rigid body dynamics. Since kinematic equations do not consider the forces that generate the motion, the resulting designed trajectory is difficult to implement directly in practice. In~\cite{3dtraj}, the authors presented an energy model with rigid body dynamics of a quadrotor UAV for generating an energy-efficient UAV communication trajectory. However, their work considers an obstacle-free environment and is not suitable for urban scenarios where the UAV needs to avoid infrastructure and trees, as shown in Figure~\ref{fig:urban_env_scene}. In our earlier work~\cite{sandhu2024scp}, a minimum-time trajectory optimization framework for a 6-DoF quadrotor UAV was developed using successive convexification techniques, demonstrating the suitability of sequential convex programming (SCP) for solving highly nonlinear UAV optimal control problems involving full rigid-body dynamics. However, communication-aware mission constraints were not considered in that formulation.

In this work, we generate a collision-free, energy-optimal trajectory for a quadrotor UAV operating in an urban environment while relaying data to a ground station. Unlike conventional trajectory planning approaches that optimize only geometric path length or mission time, the proposed formulation simultaneously considers:
\begin{itemize}
    \item full rigid-body quadrotor dynamics,
    \item cumulative communication throughput requirements,
    \item realistic urban wireless channel models,
    \item obstacle avoidance constraints, and
    \item energy-efficient trajectory generation.
\end{itemize}

Instead of using a point mass model, we consider the full rigid-body dynamics of the UAV. The trajectory generation problem is formulated as a free-time optimal control problem (OCP) to minimize the quadrotor's energy consumption. By ``free-time,'' we mean that the final time of the OCP is not predefined; rather, the mission duration itself is an optimization variable. The resulting nonconvex OCP is solved iteratively using sequential convex programming (SCP)~\cite{MAO20174063}. The SCP framework enables computationally tractable generation of dynamically feasible and communication-aware trajectories while providing strong convergence characteristics for highly nonlinear constrained optimal control problems.

The primary contributions of this work are summarized as follows:
\begin{itemize}
    \item Formulation of a communication-constrained energy-optimal UAV trajectory generation problem incorporating full rigid-body dynamics, realistic urban communication models, and obstacle avoidance constraints.
    
    \item Development of an SCP-based solution methodology for solving the resulting nonconvex free-final-time optimal control problem.
    
    \item Comprehensive numerical validation in both obstacle-free and urban environments demonstrating the ability of the proposed framework to generate dynamically feasible and communication-aware trajectories under varying data throughput requirements.
\end{itemize}

The remainder of this paper is organized as follows. Section~\ref{sec:problem} formulates the energy-optimal trajectory generation problem, considering UAV dynamics, communication constraints, and obstacle avoidance. Section~\ref{sec:convex_sub_prob} presents the solution methodology based on successive convexification. Section~\ref{sec:simulation} validates the approach through numerical simulations in various scenarios. Finally, Section~\ref{sec:conclusion} concludes the paper and discusses future research directions.

\section{Problem Formulation}
\label{sec:problem}

This section presents the mathematical framework for the communication-constrained energy-optimal trajectory generation problem. The objective is to generate dynamically feasible UAV trajectories that simultaneously minimize energy consumption, satisfy cumulative communication throughput requirements, and ensure safe navigation through obstacle-rich urban environments.

Unlike conventional trajectory planning problems that primarily optimize geometric path length or mission time using simplified kinematic models, the present formulation incorporates:
\begin{itemize}
    \item full rigid-body quadrotor dynamics,
    \item communication-aware mission constraints,
    \item realistic urban wireless channel models,
    \item obstacle avoidance requirements, and
    \item free-final-time optimal control.
\end{itemize}

The resulting problem is highly nonlinear and nonconvex due to the coupled translational and rotational vehicle dynamics, nonlinear communication-throughput constraints, obstacle avoidance requirements, and free-final-time formulation. The proposed framework integrates these interacting components within a unified optimal control framework suitable for practical autonomous UAV operations in urban environments.

\subsection{Vehicle Model}

\begin{figure}
 \centerline{
 \includegraphics[width=0.7\columnwidth]{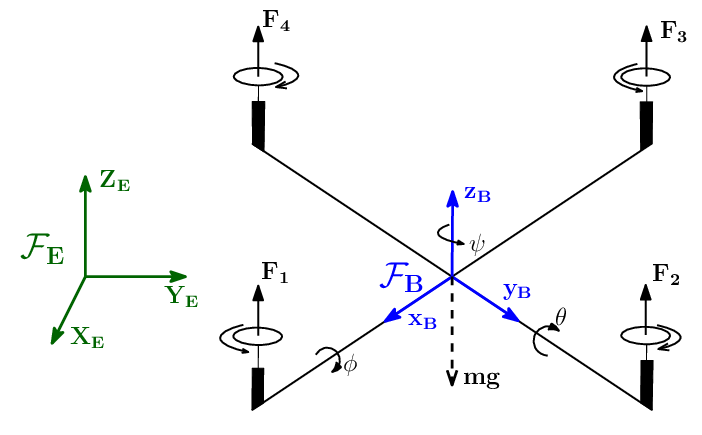}}
 \caption{A schematic diagram of a conventional quadrotor}
 \label{fig:quad_schematic}
\end{figure}

The dynamics of a quadrotor UAV, illustrated in Fig.~\ref{fig:quad_schematic}, are governed by coupled nonlinear rigid-body equations of motion. Similar modeling approaches have been extensively investigated in earlier works involving quadrotor flight dynamics, variable-pitch UAV systems, and tailsitter aerial platforms~\cite{shastry2018bemt,bhargavapuri2019vpq,swarnkar2018tailsitter}. 

The UAV position in the Earth-fixed frame $\mathcal{F}_{\text{E}}$ is denoted by
\[
\mathbf{r}(t) = [x(t), y(t), z(t)]^\top,
\]
while its orientation is represented using Euler angles
\[
\boldsymbol{\eta}(t) = [\phi(t), \theta(t), \psi(t)]^\top,
\]
corresponding to roll, pitch, and yaw, respectively. The body-fixed frame $\mathcal{F}_{\text{B}}$ is centered at the UAV center of mass, with its axes aligned with the principal inertia directions.

The nonlinear equations of motion are given by~\cite{Beard}
\begin{equation}
\label{eqn:vehicle_dynamics}
\begin{aligned}
\dot{\mathbf{r}}(t) &= \mathbf{v}(t), \\
\dot{\mathbf{v}}(t) &= -g\mathbf{e}_3 +
\frac{1}{m}\mathbf{R}(\boldsymbol{\eta}(t))\mathbf{F}(t), \\
\dot{\boldsymbol{\eta}}(t) &= \mathbf{W}(\boldsymbol{\eta}(t))\boldsymbol{\omega}(t), \\
\dot{\boldsymbol{\omega}}(t) &=
\mathbf{J}^{-1}
\left(
\boldsymbol{\tau}_c(t)
-
\boldsymbol{\omega}(t)\times
\mathbf{J}\boldsymbol{\omega}(t)
\right).
\end{aligned}
\end{equation}

Here, $\mathbf{v}(t)$ denotes the translational velocity vector, $g$ represents the gravitational acceleration, and $m$ denotes the UAV mass. The matrix $\mathbf{R}(\boldsymbol{\eta}(t))$ represents the rotation matrix transforming vectors from the body-fixed frame $\mathcal{F}_{\text{B}}$ to the Earth-fixed frame $\mathcal{F}_{\text{E}}$. The vector $\mathbf{F}(t)$ denotes the total thrust generated by the four rotors, while $\mathbf{W}(\boldsymbol{\eta}(t))$ maps body angular rates to Euler angle rates. Furthermore, $\mathbf{J}$ denotes the inertia matrix, $\mathbf{e}_3 = [0,0,1]^\top$ is the unit vector along the vertical axis of the Earth-fixed frame, and $\boldsymbol{\tau}_c(t)$ denotes the control torque vector.

The thrust generated by the $i^{\text{th}}$ rotor is modeled as
\[
F_i(t) = C_T \Omega_i^2(t),
\]
where $C_T$ denotes the aerodynamic thrust coefficient and $\Omega_i(t)$ represents the rotor angular speed. The control torques about each body axis are generated through differential thrust inputs and depend on the quadrotor geometry, where $l$ denotes the arm length.

The adopted rigid-body model captures the coupled translational and rotational dynamics of the UAV, including gyroscopic coupling effects and nonlinear attitude-motion interactions. Such a formulation provides a realistic and dynamically consistent framework for communication-aware trajectory optimization and enables generation of trajectories that remain directly implementable on practical aerial platforms.

\subsection{Communication Channel Model}

As shown in Fig.~\ref{fig:communication}, in urban environments, the communication link between a UAV and a GS can be characterized by both direct (line-of-sight) and indirect (reflected) paths.  The probability of establishing a direct communication link ($\mathbb{P}_{\text{dir}}(t)$) in this environment is modeled using a sigmoid function as~\cite{OptimalLAP}:

$$
\mathbb{P}_{\text{dir}}(t) = \frac{1}{1 + a_1 \exp(-a_2 [\alpha(t) - a_1])}
$$
where, $a_1$ and $a_2$ are parameters that depend on the specific urban environment. The elevation angle $\alpha(t)$ between the UAV and GS is given by: 

$$
\alpha(t) = \frac{180}{\pi} \sin^{-1}\left(\frac{z(t) - z_{\text{GS}}}{\Delta \mathbf{r}(t)}\right)
$$

Here, $z(t)$ and $z_{\text{GS}}$ are the altitudes of the UAV and the GS, respectively, while $\Delta \mathbf{r}(t) \coloneqq \norm{\r(t)-\r_{\text{GS}}} $ denotes the Euclidean distance between them. The communication link quality is characterized by the channel coefficient $\mathcal{C}(t) = \sqrt{\gamma(t)}\tilde{\mathcal{C}}(t)$. 
This coefficient depends on  the large-scale channel power gain $\gamma(t)$, and small-scale fading, represented by the random variable $\tilde{\mathcal{C}}(t)$, with an expected squared magnitude of unity $\br{\mathbb{E}\bsq{|\tilde{\mathcal{C}}(t)|^2 =1}}$. Also,
the channel power gain, $\gamma(t)$, varies depending on the link type:

$$
\gamma(t) = \begin{cases}
    \gamma_0 \Delta\mathbf{r}^{-\beta}(t) & \text{Direct link} \\
    \zeta \gamma_0 \Delta\mathbf{r}^{-\beta}(t) & \text{Indirect link}
\end{cases}
$$
where $\gamma_0$ is the reference channel power at 1 meter distance, the path loss exponent $\beta$ characterizes the rate at which signal strength decreases with distance, and $\zeta \in (0,1)$ represents additional attenuation due to reflections in the indirect path.

    \begin{figure}
 \centerline{
 \includegraphics[width=0.65\columnwidth]{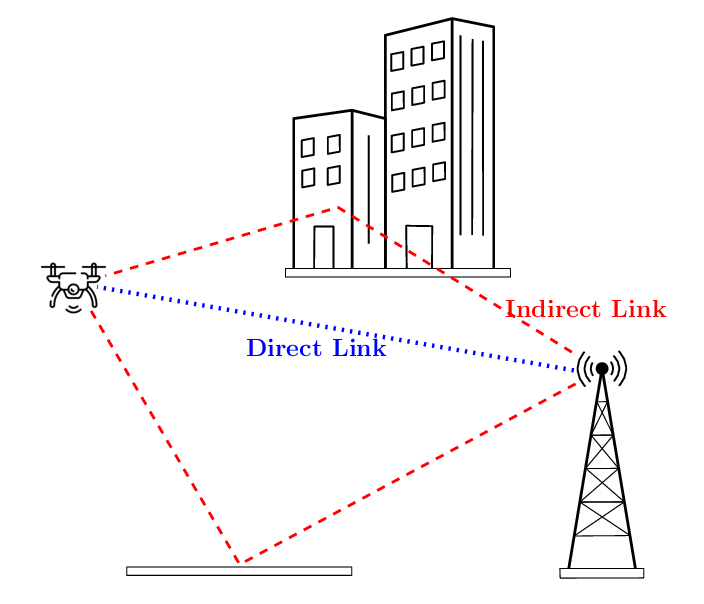}}
 \caption{{Illustration of different types of communication links in an urban city}}
 \label{fig:communication}
 \end{figure}

Given that the probability of an indirect link is complementary to that of a direct link ($\mathbb{P}_{\text{indir}} = 1 - \mathbb{P}_{\text{dir}}$), we can express the expected channel power gain as:

\begin{align}
\mathbb{E}[|\mathcal{C}(t)|^2] &= {\mathbb{P}}_{\text{dir}}\gamma_0 \Delta\mathbf{r}^{-\beta}(t) + {\mathbb{P}}_{\text{indir}} \zeta \gamma_0 \Delta\mathbf{r}^{-\beta}(t) \nonumber\\ &= \overline{\mathbb{P}}_{\text{dir}} \gamma_0 \Delta\mathbf{r}^{-\beta}(t)
\end{align}
where $\overline{\mathbb{P}}_{\text{dir}} = \zeta + (1 - \zeta) \mathbb{P}_{\text{dir}}(t)$ is the effective direct link probability.
Then, the instantaneous achievable communication rate between the UAV and ground station can be calculated using the Shannon-Hartley theorem ~\cite{ShanonHartley}:

\begin{equation}
\mathcal{R}(t) = w \log_2 \left(1 + \frac{P_{\text{tx}} |\mathcal{C}(t)|^2}{\sigma^2 \kappa_0}\right)
\end{equation}
where, $w$ denotes the bandwidth, $P_{\text{tx}}$ is the transmission power, $\sigma^2$ represents the noise power at the receiver, and $\kappa_0 > 1$  is the gap to capacity accounting for practical coding schemes.

Finally, the accumulated communication data throughput $\mathcal{Q}(t)$ over time can be quantified as:

\begin{align}
\mathcal{Q}(t) &= \int_0^t \mathbb{E}[\mathcal{R}(\tau)] \, d\tau \leq \int_0^t w \log_2 \left(1 + \frac{P_{\text{tx}}\mathbb{E}[|\mathcal{C}(\tau)|^2]}{\sigma^2 \kappa_0}\right) \, d\tau \nonumber  \\
&= \int_0^t w \log_2 \left(1 + \frac{b \overline{\mathbb{P}}_{\text{dir}}(\tau)}{\Delta\mathbf{r}^\beta}\right) \, d\tau
\end{align}

where the inequality follows from Jensen's inequality applied to the concave $\log_2(\cdot)$ function, i.e.\ $\mathbb{E}[\log_2(1+X)] \leq \log_2(1+\mathbb{E}[X])$, and $b = \frac{P_{\text{tx}}\gamma_0}{\sigma^2\kappa_0}$ consolidates several system parameters into a single term. This upper bound provides a tractable expression for trajectory optimization while maintaining accuracy for typical urban scenarios.

\subsection{Optimal Control Problem Formulation}


The trajectory generation problem is formulated as a free-time optimal control problem to minimize the quadrotor's total energy consumption over the mission time $T$, subject to dynamics, actuator limits, boundary conditions, communication throughput requirements, and obstacle avoidance constraints:

\begin{subequations}\label{eqn:min_energy_prob} 
\begin{align} 
    \min_{\u(t), T} ~& J(\x(t),\u(t),T) = \int_0^T \u(t)^\top \u(t) dt \\
    \text{s.t.} ~ 
    & \dot{\x} = f(\x(t), \u(t)), ~ t\in [0,T] \label{dynamics} \\ 
    & 0 \leq u_1(t) \leq u_{1\max}, ~ t\in [0,T] \label{control_constraint1}\\ 
    & |u_i(t)| \leq u_{i\max}, ~ i = 2,3,4, ~ t\in [0,T] \label{control_constraint2} \\
    & |\phi(t)| \leq \phi_{\max}  \label{state_constraint1}\\
    & |\theta(t)| \leq \theta_{\max} \label{state_constraint2}\\
    & \x(0) = \x_{\text{start}}  \label{boundary_constraint1}\\ 
    & \x(T) = \x_{\text{goal}} \label{boundary_constraint2}\\ 
    & \cQ(T) \geq \cQ_{\min}  \label{nonconvex_constraint1}\\ 
    &\norm{x(t)- x_{\text{o},j}, y(t)- y_{\text{o},j}} \geq d_{\text{s},j}, ~ \forall j, ~ t\in [0,T]\label{nonconvex_constraint2}
\end{align}
\end{subequations}
 The cost function, \( J \), represents the energy expenditure, which is a function of the control inputs.  The quadrotor dynamics are governed by the nonlinear system \( \dot{\x}(t) = f(\x(t), \u(t)) \) \eqref{dynamics}, where \( \x(t) \coloneqq [\r^\top, \v^\top, \boldsymbol{\eta}^\top, \dot{\boldsymbol{\eta}}^\top]^\top\in \Rn^{12}\) and \( \u(t) \coloneqq [\mathbf{F}, \boldsymbol{\tau}_c^\top]^\top\in \Rn^{4} \) are the state and control vectors, respectively.  Control constraints ensure that the vertical thrust remains non-negative and within the actuator limits \( 0 \leq u_1(t) \leq u_{1,\max} \) \eqref{control_constraint1}, while torque control inputs are bounded by \( |u_i(t)| \leq u_{i,\max} \), \( i = 2, 3, 4 \) \eqref{control_constraint2}, to reflect physical limitations.

State constraints include initial and terminal boundary conditions, ensuring that the quadrotor starts from \( \x(0) = \x_{\text{start}} \) \eqref{boundary_constraint1} and reaches a predefined final state \( \x(T) = \x_{\text{goal}} \) \eqref{boundary_constraint2}. Furthermore, roll and pitch angles are limited by \( |\phi(t)| \leq \phi_{\max} \)  and \( |\theta(t)| \leq \theta_{\max} \) \eqref{state_constraint1}--\eqref{state_constraint2}  to maintain stable flight conditions. A performance-related constraint ensures that the accumulated throughput \( \mathcal{Q}(T) \), computed as a function of state and control variables, meets or exceeds a minimum threshold \( \mathcal{Q}_{\min} \) \eqref{nonconvex_constraint1}, ensuring successful mission performance. Additionally, obstacle avoidance constraints are incorporated to prevent collisions by enforcing a safe distance from obstacles throughout the flight \eqref{nonconvex_constraint2}.


\section{Convex Subproblem Conversion}
\label{sec:convex_sub_prob}
The OCP \eqref{eqn:min_energy_prob} is inherently non-convex due to nonlinear quadrotor dynamics \eqref{dynamics}, non-convex communication constraints \eqref{nonconvex_constraint1}, and obstacle avoidance requirement \eqref{nonconvex_constraint2}.  The non-convex problems are typically difficult to solve efficiently and reliably. Moreover, the continuous time OCP  can also be viewed as an infinite dimensional optimization problem over functional spaces. In most of the cases, the solution to such infinite dimensional optimization problems is neither available in closed form nor computationally tractable to compute in real time. Thereby, we attempt to approximate the non-convex OCP by iteratively solving a sequence of convex subproblems. Figure (\ref{fig:flowchart}) presents an overview of the sequential convex programming (SCP) algorithm. The central component of this framework, shown as the "Convexification" block, linearizes non-convexities and approximates the infinite-dimensional problem by discretizing it into a finite-dimensional representation. The detailed structure of this framework is briefed in the subsequent subsections.    
\begin{figure}[H]
 \centerline{
 \includegraphics[width=1\columnwidth]{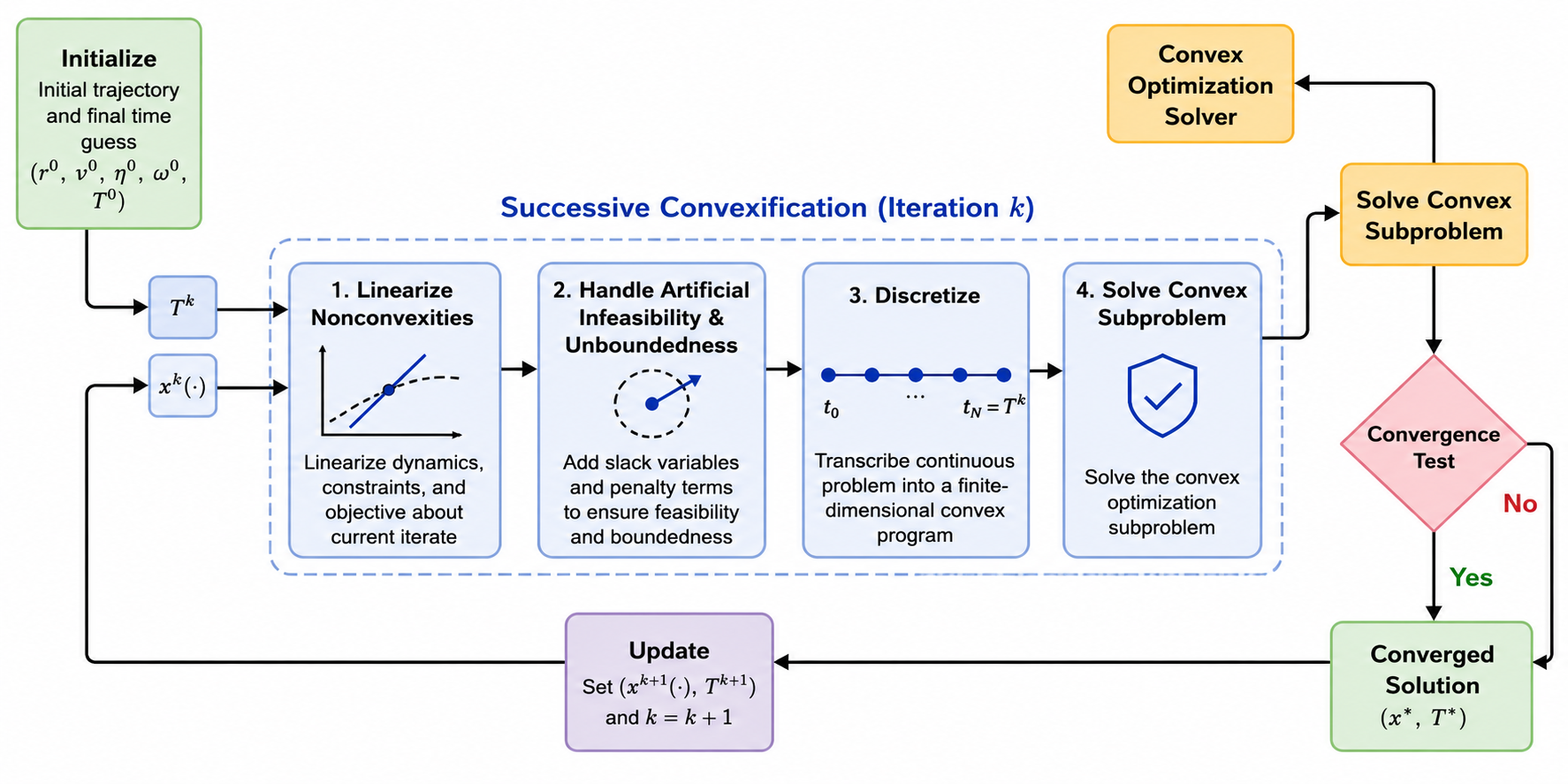}}
 \caption{{Flowchart of presented SCP framework}}
 \label{fig:flowchart}
 \end{figure}
\subsection{Time scaling}


In the free-time optimal control problem, the final time $T$ is an optimization variable rather than a fixed parameter. This variable time horizon presents numerical challenges for standard optimization solvers.  To address this challenge, we employ a time-scaling transformation that converts the problem into a fixed-time formulation. We define a normalized time variable $\tau \in [0, 1]$ that maps the original time interval $[0, T]$ to a fixed normalized interval. The relationship between the original time $t$ and the normalized time $\tau$ is given by $t = T \cdot \tau$. As $t$ goes from $0$ to $T$, $\tau$ goes from $0$ to $1$. This normalization allows for a consistent treatment of the problem regardless of the actual mission duration. The dynamics, boundary conditions, and cost function are then reformulated in terms of $\tau$ to reflect the normalized time scale. The time-scaling approach simplifies the optimization process, ensures consistent constraint enforcement, and enables the use of fixed-time optimization techniques. Hence, the dynamics  converted to normalized time is expressed as, 
\begin{equation}
   \frac{\text{d}\x}{{\text{d}}{\tau}} = \frac{\text{d}\x}{{\text{d}t}}\frac{\text{d}t}{{\text{d}}{\tau}} = {T \cdot f}\left({\x,\u}\right).
\end{equation}

The boundary conditions transform directly as 
\begin{align} 
\label{eq:33}
g_{ic}(\x(0)) = \x(0) - \x_{\text{start}}, ~~
g_{tc}(\x(1)) = \x(1) - \x_{\text{goal}}.
\end{align}

Similarly, the cost function (5a) in normalized time is given as 
\begin{align}
   J(\x(\tau),\u(\tau),T) &= h(\x(1), T)+ \int_0^1 \Gamma(\x(\tau), \u(\tau), T) \text{d}\tau \nonumber\\ 
   & = \int_0^1 T \cdot \u(\tau)^\top \u(\tau) \text{d}\tau.
\end{align} 
where the terminal cost $h(\x(1), T) \equiv 0 $ and the running cost $\Gamma(\x(\tau), \u(\tau), T) = T \cdot \u(\tau)^\top \u(\tau) $. 
Then, the OCP in the normalized time can be reformulated as, 

\begin{subequations}\label{eqn:min_energy_prob_normalized}
\begin{align}
    \min_{\u(\tau), T} ~& J(\x(\tau),\u(\tau),T) \\
    \text{s.t.} ~ 
    & \dot{\x} =  \overline{f}(\x(\tau), \u(\tau), T), ~ \tau\in [0,1]  \label{dynamics_normalized}\\ 
    & \u(\tau) \in \mathcal{U}, \quad \x(\tau) \in   \mathcal{X} \\ 
    & \x(0) = \x_{\text{start}}, ~ \x(1) = \x_{\text{goal}} \\
    & s(\x(\tau),\u(\tau),T) \leq 0 \label{nonconvex_constraints_combined}
\end{align}
\end{subequations}

Here, $\overline{f}(\x(\tau), \u(\tau), T) \coloneqq T\cdot f(\x(\tau), \u(\tau))$, and $\mathcal{U}, \mathcal{X}$ define the set containing the convex control and state constraints earlier defined in \eqref{control_constraint1}--\eqref{state_constraint1}. Also, the nonconvex constraints \eqref{nonconvex_constraint1}--\eqref{nonconvex_constraint2} in the original OCP \eqref{eqn:min_energy_prob}, are defined in \eqref{nonconvex_constraints_combined}.

\subsection{Initial Trajectory guess}
The successive convexification algorithm requires an initial reference trajectory to begin the iterative process. For simplified implementation, the initial guess  trajectory is constructed by linearly interpolating between the two boundary states:

\begin{equation}
  \x^\text{g}(\tau) = (1 - \tau)\x_{\text{start}} + \tau \x_{\text{goal}}, \, \, \, \text{for } \tau \in [0,1]. \label{5}
\end{equation}

Similarly, the control input guess is initialized as
\begin{equation}
    \u^\text{g}(\tau) = (1 - \tau)\u_0^{\text{g}} + \tau \u_f^{\text{g}}, \, \, \text{for } \tau \in [0,1], \label{6}
\end{equation}
where \( \u_0^{\text{g}} \) and \( \u_f^{\text{g}} \) are the initial and final guesses for the control inputs. Also,  $T^\text{g}$ is considred as the initial guess for the mission time.  This initialization strategy provides a dynamically infeasible but geometrically reasonable starting point that the algorithm can refine through successive iterations.

 While the algorithm demonstrates robust convergence properties from various initializations (detailed convergence analysis in Section \ref{sec:convergence}), a well-chosen initial guess can significantly reduce the number of iterations required. Specifically, an accurate starting estimate can reduce the number of iterations required for convergence, enhance the optimality of the resulting trajectory, and improve the overall viability of the converged solution.

\subsection{Linearization }
The first step in constructing a convex sub-problem involves addressing the inherent non-convexities in the OCP \eqref{eqn:min_energy_prob_normalized}. By replacing each non-convex term with its first-order approximation around the reference trajectory, the resulting sub-problem is guaranteed to be convex. Compared to second-order approximations, first-order linearizations offer a computational advantage, requiring significantly less processing effort. Hence, the corresponding Jacobian matrix ($A(\tau)$) at each iteration is calculated as,
\begin{align*}
{A}{(}{\tau}{)} & {\text{ = }}\,{\nabla}_{\x}{\overline{f}}\br{\x^r(\tau), \u^r(\tau), T^r} 
\end{align*}
where $\br{\x^r(\tau), \u^r(\tau), T^r}$ are the reference states, control inputs, and time at a particular iteration. The other Jacobians $, B(\tau) = \nabla_\u \overline{f}, F(\tau) = \nabla_T \overline{f}, C(\tau) = \nabla_\x {s}, D(\tau) = \nabla_\u {s}, G(\tau) = \nabla_T {s}$ can be similarly computed about the reference trajectories. The residues of nonlinear functions ($\overline{f}, s$) after linearization are calculated as
\begin{align*}
    \varrho{(}\tau{)} & \coloneqq\,{\overline{f}}\br{\x^r(\tau), \u^r(\tau), T^r}{-}{A}\x^r{(}\tau{)}{-}{B}{\u^r}{(}\tau{)}{-}{F}T^r \\ 
{\widehat{\varrho}}{(}\tau{)} & \coloneqq\,{s}\br{\x^r(\tau), \u^r(\tau), T^r}{-}{C}\x^r{(}\tau{)}{-}{D}{\u^r}{(}\tau{)}{-}{G}T^r
\end{align*}

However,  linearization  may introduce two fundamental challenges: artificial un-boundedness and artificial infeasibility. 
 Even when the initial (non-linear) constraints allow for a non-empty feasible set, two scenarios may arise in which the constraints become inconsistent after linearization:

\begin{itemize}
    \item The intersection of the linearized constraints may be empty. 
  
    \item The trust region~\cite{MAO20174063} may be so small that it restricts the solution variables to a region outside the feasible set. Consequently, the  intersection of the feasible area of the linearized constraints with the trust region may also be empty.  
    \end{itemize}

 To address this issue, we introduce slack variables often referred to as virtual controls $\nu(\tau)$ for dynamics constraints and virtual buffer zones $\nu_s(\tau)$ for nonconvex constraints.  The modified problem can be expressed as: 
\begin{subequations}
\label{eqn:J_lambda}
\begin{align} 
& \mathop{\min}\limits_{\u(\tau),T,\widehat{\nu}} J_{\lambda}( \x(\tau), \u(\tau), T, \widehat{\nu}), \label{9a} \\ 
& \quad \text{s.t. } \dot{\x}(\tau) = A\x(\tau) + B\u(\tau) + FT + \varrho(\tau) +\nu(\tau), \label{9b} \\ 
& \qquad \u(\tau) \in \mathcal{U},  \qquad \x(\tau) \in   \mathcal{X} \\ 
    & \qquad \x(0) = \x_{\text{start}}, ~~ \x(1) = \x_{\text{goal}} \\
& \qquad C\x(\tau) + D\u(\tau) + GT + \widehat{\varrho}(\tau) \leq \nu_s(\tau), \label{9f} \\ 
& \qquad \left\Vert \delta\x (\tau)\right\Vert + \left\Vert \delta\u (\tau)\right\Vert + \left\Vert \delta T \right\Vert\leq \Upsilon \label{eqn:trust_region_constraint}  
\end{align}
\end{subequations}
where $\widehat{\nu}$ is a shorthand notation used to denote the collection of virtual controls $(\nu, \nu_s)$. These virtual controls are auxiliary variables that do not correspond to physical control inputs of the UAV system. Rather, they are synthetic inputs introduced solely to ensure numerical stability and feasibility during the optimization process. In the final converged solution, these virtual controls $\widehat{\nu}(t)$ must vanish (i.e., equal zero), indicating that the obtained trajectory satisfies all original constraints without requiring artificial relaxations. To prevent artificial un-boundedness, we impose trust region constraints \eqref{eqn:trust_region_constraint}. The objective function \eqref{9a} incorporating penalties on the virtual  variables with the original objective function from \eqref{eqn:min_energy_prob_normalized} is expressed as 
\begin{subequations}
\begin{align}
J_{\lambda}(\x, \u, T, \widehat{\nu}) = 
&  \int_{0}^{1} \Gamma_{\lambda}(\x, \u, T, \nu, \nu_s) \, \text{d}\tau, \label{14} \\
\Gamma_{\lambda}(\x, \u, T, \nu, \nu_s) = & \Gamma(\x, \u, T) + \lambda P(\nu, \nu_s), \label{15}
\end{align}
\end{subequations}

where  \(P\) is a positive definite penalty function defined as:

\begin{equation}
P(y, z) = \left\Vert y \right\Vert_1 + \left\Vert z \right\Vert_1. \label{eqn:penalty_func} 
\end{equation}

\subsection{Discretization}

The optimal control problem \eqref{eqn:J_lambda}, after linearizing, is convex but   remains continuous in time. To make it numerically tractable,  we  discretize it by selecting $N$ evenly spaced temporal nodes:
\begin{align*}
&0 = \tau_{1}\,{<}\,{\cdots}\,{<}\,\tau_{N} \text{ = } {1} , ~ \nonumber\\ &\Delta \tau \coloneqq \tau_{k+1} - \tau_{k} = 1/(N-1) \,\,\, \forall \, k = 1, \ldots, N.  
\end{align*}

Zero-order hold (ZOH) is applied to the control inputs, keeping $\u = \u_k$ constant over each interval $\tau \in \bsq{\tau_k, \tau_{k+1}}$.  The  dynamics relating ${\x}_k$ to ${\x}_{k+1}$ is obtained through exact integration of the linearized system,
\begin{align}
{\dot{{\x}}}{(}\tau{)} = & {A}{(}\tau{)}{\x}{(}\tau{)} + {B}{(}\tau{)}{\u}_{k}  {+}{F}{(}\tau{)}T + {\varrho}{(}\tau{).} + \nu(\tau) \label{16} 
\end{align}

The solution for the above linear time varying system is given by:
\begin{align*}{\x}{(}\tau{)} =&  \Phi{(}\tau,\tau_{k}{)}{\x}{(}\tau_{k}{)} \\&\qquad{+}\mathop{\int}\nolimits_{\tau_{k}}\nolimits^{\tau}{\Phi}{(}\tau,{t}{)[}{B}{(}{t}{)}{\u}_{k}   {+}{F}{(}{t}{)}{T} + {\varrho}{(}{t}{)}+\nu(t){]}{\text{d}}{t} , 
\end{align*}
where ${\Phi}{(}\tau,\tau_{k}{)}$ is the state transition matrix  \cite{ogata} satisfying,
\begin{equation}
    {\dot{\Phi}}{(}\tau,\tau_{k}{)} = {A}{(}\tau{)}{\Phi}{(}\tau,\tau_{k}{)},~{\Phi}{(}\tau_{k},\tau_{k}{)} = {I}_{n}{.} 
\end{equation}
By assigning $\tau = \tau_{k+1}$, \eqref{16} forms a linear time-varying  difference equation, which advances the discrete-time state $\x_k$  to the subsequent discrete-time state $\x_{k+1}$,
    \begin{align}{\x}_{{k} + {1}} & = {A}_{k}{\x}_{k} + {B}_{k}{\u}_{k}  + {F}_{k}{T} + {\varrho}_{k} + \nu_k, \label{S32a} 
        \end{align}
    where the matrices are given by
    \begin{align*}
        {A}_{k} & = {\Phi}{(}\tau_{{k} + {1}},\tau_{k}{)},  \\ 
    {B}_{k} & = {A}_{k} \mathop{\int}\nolimits_{\tau_{k}}\nolimits^{\tau_{{k} + {1}}}{\Phi}{(}{t},\tau_{k}{)}^{{-}{1}}{B}{(}{t}{)}{\text{d}}{t},
    \end{align*}
and $F_k,\rho_k, \nu_k$ can be computed as that of $B_k$. 
Then, the OCP after discretization is formulated as follows, 
\begin{subequations} \label{eqn:discretized_convex_subproblem}
\begin{align} 
& \mathop{\min}\limits_{\uu,T,\widehat{{\uv}}}{\cal{L}}_{\lambda}{(}{\ux},{\uu},{T},\widehat{{\uv}}{)}, \label{13a} \\ 
& \quad {\text{s.t. }}{\x}_{{k} + {1}} = {A}_{k}{\x}_{k} + {B}_{k}{\u}_{k}  + {F}_{k}{T} + {\varrho}_{k} + {\nu}_{k}, \label{13b} \\ 
&\qquad \u_k \in \mathcal{U}, \qquad \x_k \in \mathcal{X}\\
&\qquad \x_1 = \x_{\text{start}}, ~ \x_N = \x_{\text{goal}} \\
& \qquad {C}_{k}{\x}_{k} + {D}_{k}{\u}_{k} + {G}_{k}{T} + {{\widehat{\varrho}}}_{k}\,\le{\nu}_{s,k}, \label{13f} \\ 
& \qquad {\left\Vert{{\delta}{\x}_{k}}\right\Vert} + {\left\Vert{{\delta}{\u}_{k}}\right\Vert} + {\left\Vert{{\delta}{T}}\right\Vert}\le\,{\Upsilon}{.} \label{13i} 
\end{align}
\end{subequations}
where $\ux, \uu, \widehat{\uv}$ are defined  as the \emph{super vectors} containing the states ($\x_k$), controls ($\u_k$), and virtual variables ($\hat{\nu}_k$) of all the temporal nodes, and The discretized cost function is defined as,  
 \begin{align}{\cal{L}}_{\lambda}{(}{\ux},{\uu},{T},\widehat{{\uv}}{)} & =  {\text{ trapz}}{(}{\Gamma}_{\lambda}^{N}{)},  \label{19a} 
 \end{align} 
where ${\Gamma}_{{\lambda},{k}}^{N}  = {\Gamma}_{\lambda}{(}{\x}_{k},{\u}_{k},{T},{\nu}_{k},{\nu}_{s,k}{)}$ and trapz(.) is trapezoidal numerical integration defined as,
  \begin{equation}
    {\text{trapz}}{}{(}{d}{)}\,{=}\,\frac{{\Delta}{\tau}}{2}\mathop{\sum}\limits_{{k} = {1}}\limits^{{N}{-}{1}}\left({d}_{k} + {d}_{{k} + {1}}\right){.} \label{20}
\end{equation} 


\subsection{Update Rule and Stopping Criterion}
The successive convexification algorithm requires mechanisms to assess solution quality and determine convergence. At each iteration, we compare the actual improvement in the nonlinear cost function with the predicted improvement from the convex approximation. We
define the nonlinear discretized cost function as 
\begin{align*}
    \widehat{J}_\lam(\ux, \uu, T) \coloneqq \text{trapz}\br{\widehat{\Gamma}_\lam^N}, 
\end{align*}
where $\widehat{\Gamma}_{\lam,k}^N \coloneqq {\Gamma}_{\lambda}\br{{\x}_{k},{\u}_{k},{T}, \delta_k, \bsq{s(\tau_k, \x_k,\u_k,T)}_+ }$. Here $\bsq{a}_+ \coloneqq \max\bc{0,a}$ and $\delta_k$ is the actual dynamics defect due to linearization, defined as
\begin{equation*}
    \delta_k \coloneqq \x_{k+1} - \br{A_k \x_k + B_k \u_k + F_k T + \varrho_k}
\end{equation*}
where  $\widehat{f}$ is the discretized version of the nonlinear system dynamics \eqref{dynamics_normalized}. 
The algorithm uses both linear as well as nonlinear augmented cost functions to compute the accuracy of convex approximation by a parameter $\rho$ given as,

\begin{align}
\rho&=\frac{{\widehat{J}}_{\lambda}{(}\ux^r,{\uu^r},{T^r}{)}{-}{\widehat{J}}_{\lambda}{(}{\ux}^{\star},{\uu}^{\star},{T}^{\star}{)}}{{\widehat{J}}_{\lambda}{(}\ux^r,{\uu^r},{T^r}{)}{-}{\cal{L}}_{\lambda}{(}{\ux}^{\star},{\uu}^{\star},{T}^{\ast},{\widehat{\uv}}^{\star}{)}} \notag\\
\implies {\rho}\,&{\text{ = }}\,\frac{J^r{-}J^s}{J^r{-}{\cal{L}}^s},\label{17} 
\end{align}
where $J^r, J^s, \cL^s $ are the shorthand notations used to denote the respective expressions. This ratio $\rho$ measures how well the linear approximation predicts the actual cost reduction.
If $\rho < 0$, then $J^r$ is less than $J^s$ i.e. the solution worsens the objective. Hence the solution is rejected and the trust region is reduced.  If $\rho$ is close to 1, then the linearized cost of the solution is close to the non-linearized cost of the solution. Thereby, accept the solution and possibly expand trust region.  Thus, the solution of current iteration is passed as a reference trajectory to the next iteration. 

The stopping criterion's main objective is to assess how different the new solution trajectory is from the reference trajectory at each iteration.  Hence, by taking a sufficiently small value $\epsilon$, we measure the gap between the reference cost function and the current cost. Therefore the terminating condition is verified by,
\begin{equation}
J^r - \cL^s \leq \epsilon. \label{eqn:stopping_criteria}
\end{equation}

This criterion ensures that the linear approximation accurately represents the nonlinear problem in the neighborhood of the solution, indicating that further iterations will yield negligible improvement. The complete successive convexification procedure for solving the communication-constrained trajectory optimization problem is briefed in Algorithm \ref{alg:uav_trajectory}.

\begin{algorithm}
\caption{Communication Constrained Optimal Trajectory} 
\label{alg:uav_trajectory}
\begin{algorithmic}[1]
\Require  Penaly weight $\lam > 0$, tolerance $\epsilon>0$, Maximum number of iterations "\text{iter\_max}", Number of temporal discretizations points $N$,  update parameters $0<\rho_0<\rho_1<\rho_2<1$, and $\alpha>0$.   
\State \textbf{Initialization:} A guess UAV trajectory $\ux^g$, control input $\uu^g$,  final time $T^g$, and trust radius $\Upsilon_0$
\State \quad Set $\ux^r = \ux^g$, $\uu^r = \uu^g$ and $T^r = T^g$ 
\While{$i \leq \text{iter\_max}$ \textbf{and} stopping criteria \eqref{eqn:stopping_criteria} not met} \label{algo:step_while}
    \State Solve the convex optimal control subproblem \eqref{eqn:discretized_convex_subproblem}:
    \Statex \quad \quad Obtain ${\uu}^\star$, $T^\star$ and  Update UAV trajectory $\ux^\star$
    \State Calculate $\rho$ using \eqref{17} \label{algo:step_cal_rho}
    \If{$\rho <\rho_0$} \label{algo:step_contract}
    \Statex \qquad Contract the trust region ($\Upsilon^k \leftarrow \Upsilon^k/\alpha$), and go back to Step-\ref{algo:step_while}
    \Else
    \Statex \qquad Update $\ux^r = \ux^\star$, $\uu^r = \uu^\star$, $T^r = T^\star$, and
    \Statex \begin{align*}
        \Upsilon^{k+1} = \begin{cases}
            \Upsilon^k/\alpha, ~ &\text{for}~ \rho<\rho_1\\
            \Upsilon^k, ~ &\text{for}~ \rho_1\leq\rho<\rho_2\\
            \alpha\Upsilon^k, ~ &\text{for}~ \rho\geq\rho_2
        \end{cases}
    \end{align*}
    \EndIf
    \State $i \gets i + 1$
\EndWhile
\State Return $\ux^\star$, $\uu^\star$, $T^\star$   
\end{algorithmic}
\end{algorithm}

\section{Convergence Analysis}
\label{sec:convergence}
In this section, we present the   convergence analysis of the proposed algorithm (Algorithm \ref{alg:uav_trajectory}) in discrete time setting. Therefore, first we reformulate the normalized non-convex OCP \eqref{eqn:min_energy_prob_normalized}. To do that, we treat  all discrete states ($\ux$), controls ($\uu$), and parameter ($T$) as independent variables. Hence, the reformulated non-convex discrete time problem is 

\begin{subequations}\label{eqn:J_dis}
\begin{align}
    \min_{\ux, \uu, T} ~& \widehat{J}(\ux, \uu, T) = \text{trapz}\br{\widehat{\Gamma}_\lam^N} \\
    \text{s.t.} ~ 
    & \x_{k+1} = \widehat{f}(\x_k, \u_k, T) , \quad \forall k\in\{ 1,\ldots, N \}\\ 
    & \u_k \in \mathcal{U}, \quad \forall k\in\{ 1,\ldots, N \} \\ 
    & \x_k \in   \mathcal{X}, \quad \forall k\in\{ 1,\ldots, N \} \\     
    & \x_1 = \x_{\text{start}}, ~ \x_N = \x_{\text{goal}} \\
    & s(\x_k,\u_k,T) \leq 0 , \quad \forall k\in\{ 1,\ldots, N \} 
\end{align}
\end{subequations}
where trapz(.) is numerical integration method defined in 
\eqref{20}, $\widehat{\Gamma}_\lam^N$ is the running cost 
with discrete state and control variables, and  $\widehat{f}$ is the discretized version of the nonlinear system dynamics \eqref{dynamics_normalized}. 
For simplicity, throughout the analysis, we denote our variable trio $\bsq{\ux^\top,\uu^\top,T}^\top$ as a single variable $\uz$ and at the $k$th iteration $\uz^k \coloneqq \bsq{{\ux^k}^\top,{\uu^k}^\top,T^k}^\top$. Following this notations, the system dynamics \eqref{dynamics_normalized} and the nonconvex path constraints \eqref{nonconvex_constraints_combined} are represented as a set of algebraic equations $g_i(\uz) = 0$ and $\overline{g}_j(\uz) \leq 0$ respectively with $i\in \In_{\text{eq}}$ (set of equality constraint indices) and $j\in \In_{\text{NC}}$ (set of nonconvex inequality constraint indices). Similarly, the convex state and control constraints are represented as $h_l(\uz) \leq 0$, where $l\in \In_{\text{C}}$ (set of convex inequality constraint indices). Now, the discrete time non-convex OCP \eqref{eqn:J_dis} is rewritten as, \\
\begin{subequations}\label{eqn:J_dis_Z}
\begin{align}
    \min_{\uz} ~& \widehat{J}(\uz) = \text{trapz}(\widehat{\Gamma}_\lam^N) \\
    \text{s.t.} ~ 
    & g_i(\uz) = 0, \quad i\in \In_{\text{eq}}\\ 
    &\overline{g}_j(\uz) \leq 0 \quad j\in \In_{\text{NC}}\\ 
    & h_l(\uz) \leq 0  \quad l\in \In_{\text{C}}
\end{align}
\end{subequations}
\\

Rewriting the cost function $\widehat{J}(\uz)$ and the constraints $g_i(\uz) = 0$ and $\overline{g}_j(\uz) \leq 0$ using the penalty function ($P(.)$) defined in \eqref{16}, we get the OCP, 
\\
\begin{align}
    \label{eqn:J_dis_pen_def}
    \widehat{J}_\lambda (\uz) \coloneqq \widehat{J}(\uz)  + \lambda \sum_i \norm{g_i(\uz)}_1 + \lambda \sum_j \norm{\overline{g}_j(\uz)}_1
\end{align}
\\
Therefore, the reformulated penalized OCP we will be considering throughout the convergence analysis is presented as 
\begin{subequations}\label{eqn:J_dis_Z_final}
\begin{align}
    \min_{\uz} ~& \widehat{J}_\lam(\uz)  \\
    \text{s.t.} ~ 
    & h_l(\uz) \leq 0 
\end{align}
\end{subequations}

\subsection{Exact problem formulation:}
The penalty function $P(.)$ can be  considered as "exact", if there exists a finite value of parameter $\lambda$ such that the OCP before penalizing \eqref{eqn:J_dis_Z}, and after penalizing \eqref{eqn:J_dis_Z_final}  are equivalent in the sense of optimality conditions. Hence, the solution of \eqref{eqn:J_dis_Z_final} would be the same as that of solving \eqref{eqn:J_dis_Z}. Before proceeding with the analysis, we introduce the following preliminary results that will set the foundation for the convergence proof.


\begin{lemma} \label{lem:stationary}
    \textbf{(Stationarity Conditions:)} If $\overline{\uz}$ is a local minima for the non-penalized OCP \eqref{eqn:J_dis_Z}, then there exist the positive Lagrange multipliers $\mu_i$, $\xi_j$, $\sigma_l$  such that, 
    \begin{align}
        \nabla \widehat{J}(\overline{\uz}) + \sum_i \mu_i \nabla g_i(\overline{\uz})+ \sum_j \xi_j \nabla \overline{g}_j\overline{\uz} + \sum_l \sigma_l \nabla h_l(\overline{\uz}) = 0
    \end{align}
\end{lemma}

Lemma \ref{lem:stationary} gives us the first order optimality condition for the OCP \eqref{eqn:J_dis_Z}. Now, we will examine the first order optimality condition for OCP \eqref{eqn:J_dis_Z_final}. The presence of $\norm{.}_1$ (1-norm) makes the penalized cost function $\widehat{J}_\lam$ non-smooth, hence not differentiable everywhere. However, it can be realized that the constraints $g_i(\uz)$ and $\overline{g}_j(\uz)$ are continuously differentiable. Hence, $\widehat{J}_\lam$ 
 is concluded to be Lipschitz continuous~\cite{Clarke1990}. Thereby, the generalized directional derivative (GDD) of $\widehat{J}_\lam$  calculated at $\overline{\uz}$ in any particular direction $\s$ is defined as. 
 \begin{align} \label{eq:gdd}
     \text{d}\widehat{J}_\lam (\overline{\uz}, \s) &\coloneqq  \limsup\limits_{\mathclap{\substack{\uz\rightarrow \overline{\uz} \\ \Delta\rightarrow 0^+}}} \frac{\widehat{J}_\lam(\uz+\Delta \s)-\widehat{J}_\lam(\uz)}{\Delta} \\ 
     & = \max \{ \mathfrak{v}^\top \s | \mathfrak{v} \in \partial \widehat{J}_\lam(\overline{\uz}) \}
 \end{align}
where $\partial\widehat{J}_\lam(\overline{\uz})$ is the subgradient. 
\begin{lemma} \label{lem:necessary}
    \textbf{(Necessary condition:)} If $\overline{\uz}$ is a local minima of OCP \eqref{eqn:J_dis_Z_final}, then $0 \in \Gn(\uz)$, where $\Gn(\uz) \coloneqq\{ \uz| 0 \in \partial{\widehat{J}_\lam(\uz)} + \uy\,|\,\uy=\sum_{l\in  \In_{\text{C}}}\sigma_l \nabla h_l({\uz}) \}$ defines the set of all stationary points of OCP \eqref{eqn:J_dis_Z_final} 
\end{lemma}

In the other words,  Lemma \ref{lem:necessary} implies that if $\overline{\uz}$ solves OCP \eqref{eqn:J_dis_Z_final}, then $\overline{\uz} \in \Gn$.

We also define a set of all feasible 
\begin{lemma}
    \textbf{(Exactness:)} \label{lem:exactness} \cite{Clarke1990} Considering $\overline{\uz}$ to be a stationary point of OCP \eqref{eqn:J_dis_Z} with $\mu_i, \xi_j, \sigma_l$ as the positive Lagrange multipliers, and if the penalty function weight $\lambda$ satisfies $\lambda > |\mu_i| \cup |\xi_j| \cup  |\sigma_l|$ for all $i \in \In_{\text{eq}}, j \in \In_{\text{NC}}, l \in \In_{\text{C}}$, then $\overline{\uz}$ is also a constrained stationary point for OCP \eqref{eqn:J_dis_Z_final}, and this implies $\overline{\uz} \in \Gn(\uz)$.  Also, we can say that, if $\overline{\uz} \in \Gn(\uz)$ is a feasible point for OCP \eqref{eqn:J_dis_Z}, then it is a stationary point for OCP \eqref{eqn:J_dis_Z}. 
\end{lemma}

    

\begin{rem}
     Lemma \ref{lem:exactness} does not directly giveaway an implementable value of $\lambda$, it still has important theoretical importance. During implementations, we
select a relatively large $\lambda$ and keep it fixed for the whole process.  From a practical point of view, an approximate range of $10^2$ to $10^5$ would work in most of the practical problems. If the problem is feasible, yet the algorithm is not converging the user may retry by increasing $\lam$. Also, very high value of $\lam$ may lead higher number of iteration  to converge.
\end{rem}

\begin{rem}
     It can be noted that Lemma \ref{lem:exactness} guarantees that as long as we can find a stationary point for the penalty OCP \eqref{eqn:J_dis_Z_final} which is feasible to the original OCP \eqref{eqn:J_dis_Z}, we reach stationarity of the original OCP \eqref{eqn:J_dis_Z} as well.
\end{rem}

\subsection{Convergence of Algorithm-\ref{alg:uav_trajectory}}
For simplicty in notations, we have already redefined the optimization variable trio at $k$th iteration  $(\x^k, \u^k, T)$ as a single variable $\uz^k$ at the start of the Section-\ref{sec:convergence}.
Omitting the dependencies the virtual control the linearized penalized cost function at $k$th iteration as $\cL_\lam^k$ can be rewritten as, 
 \begin{align}
 {\cal{L}}_{\lambda}^k{(}{\delta \ux^k},{\delta \uu^k},{\delta T}{)}  &=  {\text{ trapz}}{(}{{\Gamma}(\ux^k+\delta \ux^k, \uu^k+\delta \uu^k, T+\delta T )} {)} \nonumber\\ &\qquad \quad+ \sum_{i=1}^{N-1} \lambda_i P(\widehat{\uv}) \nonumber \\ 
 \implies  \cL_\lam^k (\Delta \uz^k) &= \text{trapz} \left( {\Gamma}(\uz^k+ \Delta \uz^k)\right)  + \sum_{i=1}^{N-1} \lam_i P(\widehat{\uv})
 \end{align}
where $\text{trapz}(\boldsymbol{\Gamma})$ applies the trapezoidal integration rule \eqref{20} to the discretized running cost $\widehat{\Gamma}_\lambda^N$ at all temporal nodes, and $\widehat{\mathbf{V}} = \{\boldsymbol{\nu}_k, \nu_{s,k}\}_{k=1}^N$ denotes the collection of all virtual control variables.

Thereby, at the end of $k$th iteration,  the actual change in the penalized cost function ($ \Delta \widehat{J}_\lam^k$) of OCP \eqref{eqn:J_dis_Z_final}  and the predicted change after linearization ($\Delta\cL_\lam^k $) are redefined by 

\begin{align*}
    \Delta \widehat{J}_\lam^k & \coloneqq  \widehat{J}_\lam(\uz^k) - \widehat{J}_\lam(\uz^k+ \Delta \uz^k) \\ 
    \Delta \cL_\lam^k &\coloneqq  \widehat{J}_\lam(\uz^k) -  \cL_\lam^k (\Delta \uz^k)
\end{align*}
\begin{theorem} \label{thm:stationarity_finite}
    The predicted cost reductions (step \ref{algo:step_cal_rho} of Algorithm \ref{alg:uav_trajectory}) $\Delta \cL_\lam^k \geq 0$ for all $k$ where $\Delta \cL_\lam^k = 0$ implies that $ \uz^{k} \in \Gn(\uz) $. Hence, $\uz^{k}$ is a stationary point of the OCP \eqref{eqn:J_dis_Z_final}. 
\end{theorem}
\begin{IEEEproof}
    Denote $\Delta \uz^\star$ as the solution to the convex subproblem \eqref{eqn:discretized_convex_subproblem}, and note that $\Delta \uz=0$ is always a feasible solution to said problem. Hence we have
	\begin{equation*}
	\cL_\lam^k(\Delta \uz^\star) \leq \cL_\lam^k(0) = \widehat{J}_\lam(\uz^k).
	\end{equation*}
	Therefore, $\Delta \cL_\lam^k=\widehat{J}_\lam(\uz^k)-\cL_\lam^k(\Delta \uz^\star)\geq 0 $, and $ \Delta \cL_\lam^k=0 $ holds if and only if $ \Delta \uz^\star=0 $ is the minimizer of \eqref{eqn:discretized_convex_subproblem}. When $ \Delta \cL_\lam^k=0 $, one can directly apply Lemma \ref{lem:necessary} to get $ \uz^k \in \Gn(\uz) $.
\end{IEEEproof}
\begin{theorem} \label{thm:contraction}
    Let $\overline{\uz}$ be a feasible point that is not a stationary point for OCP \eqref{eqn:J_dis_Z_final}  (i.e. $\overline{\uz} \notin \Gn(\uz)$). For any positive constant  ${\mathfrak{k}} \in (0,1),$ there exist strictly positive values $\overline{\Upsilon}$ and $\overline{\varepsilon}$ such that for all $\uz$ within the neighborhood $\cN(\overline{\uz}, \overline{\varepsilon})$, and all $\Upsilon \in \left(0, \overline{\Upsilon} \right]$, optimal solution $\Delta{\uz}^\star$ of OCP \eqref{eqn:J_dis_Z_final} when solved at $\uz$ with trust region radius $\Upsilon$ satisfies, 
\begin{equation}
	\rho(\uz,\Upsilon) \coloneqq\frac{\widehat{J}_\lam(\uz)-\widehat{J}_\lam(\uz+\Delta{\uz}^\star)}{\widehat{J}_\lam(\uz)-\cL_\lam(\Delta{\uz}^\star)} \geq \mathfrak{k}.
	\end{equation}
\\
\end{theorem}
\begin{IEEEproof}
    We begin by noting that, since $ \overline{\uz} $  is  feasible  yet not a stationary point, which means $ \overline{\uz}\notin \Gn(\uz) $.
	The subgradient $ \partial \widehat{J}_\lambda(\overline{\uz}) $  forms a closed convex set by virtue of being an intersection of half-spaces. Also, the set
	\begin{equation} \label{eq:conic}
		\left\lbrace \uy\,|\,\uy=\sum_{l\in I(\overline{\uz})}\sigma_l \nabla h_l(\overline{\uz}), \; \text{for some } \sigma_l \geq 0 \right\rbrace \triangleq \Cn(\overline{\uz})
	\end{equation}
	constitutes a closed convex cone as it represents a conic combination of the gradient vectors $ \nabla h_l(\overline{\uz}) $.  Since the Minkowski sum preserves convexity, it follows that $ \Gn(\overline{\uz}) $ is also a closed and convex set. 
    
    Applying the separation theorem for convex sets to $ \Gn(\overline{\uz}) $, we can establish the existence of a unit vector $ \s $ and a scalar $ \kappa > 0 $ such that for all vectors $ \mathfrak{v}\in \Gn(\overline{\uz}) $,
	\begin{align}
		 \langle\mathfrak{v}, \s\rangle \leq -\kappa <0 \overset{\eqref{eq:gdd}}{\implies} \max \{\langle\mathfrak{v}, \s \rangle\,|\,\mathfrak{v}\in \partial \widehat{J}_\lam(\overline{\uz})\} \leq -\kappa. \label{eq:separation}
	\end{align} 
    
     Consequently, we have
	\begin{equation*}
		d\widehat{J}_\lam(\overline{\uz},\s) \coloneqq \limsup\limits_{\mathclap{\substack{\uz\rightarrow \overline{\uz} \\ \Upsilon\rightarrow 0^+}}} \frac{\widehat{J}_\lam(\uz+\Upsilon \s)-\widehat{J}_\lam(\uz)}{\Upsilon} \leq -\kappa.
	\end{equation*}
	This relationship guarantees the existence of positive values  $ \overline{\Upsilon} $ and $ \overline{\epsilon} $ such that for every feasible $ \uz \in \cN(\overline{\uz},\overline{\epsilon}) $ and $ 0<\Upsilon\leq \overline{\Upsilon} $,
	\begin{equation} \label{eq:limit_delta}
		\frac{\widehat{J}_\lam(\uz+\Upsilon s)-\widehat{J}_\lam(\uz)}{\Upsilon} < -\frac{\kappa}{2}.
	\end{equation}
	Another subset of $ D(\overline{\uz}) $ is $ C(\overline{\uz}) $ defined in \eqref{eq:conic}, so \eqref{eq:separation} also holds for all $ \nu \in C(\overline{\uz}) $. This necessarily implies, by~\eqref{eq:separation},
	\begin{equation} \label{eq:feasible_s}
		\langle\nabla h_l(\overline{\uz}), \s\rangle \leq 0, \quad l \in I(\overline{\uz}),
	\end{equation}
	which essentially means $ \s $ is a feasible direction of the active constraints at $ \overline{\uz} $.
	
	When the OCP \eqref{eqn:discretized_convex_subproblem}  is solved at $ \uz $ with trust region radius $ \Upsilon $, obtaining an optimal solution $ \Delta{\uz}^\star $. By a first-order Taylor expansion of $\widehat{J}_\lam$ around $\uz$ (using the GDD definition \eqref{eq:gdd}), the \underline{actual} change in $ \widehat{J}_\lam $ is expressed as
	\begin{align}
		\Delta \widehat{J}_\lam(\uz,\Delta{\uz}^\star) &=\widehat{J}_\lam(\uz)-\widehat{J}_\lam(\uz+\Delta{\uz}^\star) \nonumber \\
		&=\widehat{J}_\lam(\uz)-\cL_\lam(\Delta{\uz}^\star)-o(\|\Delta{\uz}^\star\|) \nonumber \\
		&=\Delta \cL_\lam(\Delta{\uz}^\star)-o(\|\Upsilon\|). \label{eq:delta_j}
	\end{align}
	Accordingly, the ratio becomes
	\begin{equation*}
		\rho(\uz,\Upsilon)=\frac{\Delta \widehat{J}_\lam(\uz,\Delta{\uz}^\star)}{\Delta \cL_\lam(\Delta{\uz}^\star)} = 1-\frac{o(\|\Upsilon\|)}{\Delta \cL_\lam(\Delta{\uz}^\star)}.
	\end{equation*}
	 Let $ \Delta{\uz} = \Upsilon \s $. Since $ \s $ is a feasible direction (as established by \eqref{eq:feasible_s}), $ \Delta{\uz} $ constitutes a feasible solution to the OCP \eqref{eqn:discretized_convex_subproblem},  while $ \Delta{\uz}^\star $  represents the optimal solution. Therefore, we have 
	 \begin{equation}
	 	\cL_\lam(\Delta{\uz}^\star)\leq \cL_\lam(\Delta{\uz})  \implies \Delta \cL_\lam(\Delta{\uz}^\star) \geq \Delta \cL_\lam(\Delta{\uz}) \label{eq:delta_l}
	 \end{equation}
	 
	 As $ \Upsilon \rightarrow 0 $ with, $ 0<\Upsilon\leq \overline{\Upsilon} $, we derive
	 \begin{equation*}
	 	\Delta \widehat{J}_\lam(\uz,\Delta{\uz})=\widehat{J}_\lam(\uz)-\widehat{J}_\lam(\uz+\Delta{\uz}) > \left( \frac{\kappa}{2} \right) \Upsilon.
	 \end{equation*}
	 Substituting $ \Delta{\uz}^\star $ for $ \Delta{\uz} $ in \eqref{eq:delta_j},
	 \begin{equation*}
	 	\Delta \widehat{J}_\lam(\uz,\Delta{\uz})=\Delta \cL_\lam(\Delta{\uz})-o(\|\Upsilon\|) \geq \left( \frac{\kappa}{2} \right) \Upsilon.
	 \end{equation*}
	 Combining this with \eqref{eq:delta_l}, we get
	 \begin{equation*}
	 	\Delta \cL_\lam(\uz,\Delta{\uz}^\star) \geq \Delta \cL_\lam(\Delta{\uz})-o(\|\Upsilon\|) > \left( \frac{\kappa}{2} \right) \Upsilon.
	 \end{equation*}
	 This leads to $ \Delta \cL_\lam(\uz,\Delta{\uz}^\star) > ( \kappa/2 ) \Upsilon + o(\|\Upsilon\|) $, allowing us to express:
	 \begin{equation*}
	 	\rho(\uz,\Upsilon) = 1-\frac{o(\|\Upsilon\|)}{\Delta \cL_\lam(\Delta{\uz}^\star)} > 1-\frac{o(\|\Upsilon\|)}{( \kappa/2 ) \Upsilon + o(\|\Upsilon\|)}.
	 \end{equation*}
	 As $ \Upsilon \rightarrow 0 $, $ \rho(\uz,\Upsilon) \rightarrow 1 $, which means that for any constant $\mathfrak{k}<1$, we can assert $ \rho(\uz,\Upsilon) \geq \mathfrak{k} $.
\end{IEEEproof}

\begin{rem}
    An undesirable situation is one where steps are rejected indefinitely (i.e. by producing solutions that result in $\rho(\uz,\Upsilon) < \rho_0$). Theorem \ref{thm:contraction} provides an assurance that  Algorithm-\ref{alg:uav_trajectory} will not produce such behavior. By contracting $\Upsilon^k$ sufficiently, Theorem \ref{thm:contraction} guarantees that the ratio $\rho^k$ will eventually exceed $\rho_0$, and the algorithm will stop rejecting steps.
\end{rem}

\begin{theorem}
    \textbf{(Limit Point Guarantee:)} When   Algorithm-\ref{alg:uav_trajectory}  generates an infinite sequence $ \{ \uz^k \} $, this sequence  is assured to possess limit points. Moreover, any such limit point, $ \overline{\uz} $, is also a stationary point of the non-convex penalty OCP \eqref{eqn:J_dis_Z_final}.
\end{theorem}

\begin{IEEEproof}
    The convexity and compactness of the feasible domain of \eqref{eqn:J_dis_Z_final} ensures the existence of at least one convergent subsequence $ \{ \uz^{k_{i}} \} \rightarrow \overline{\uz} $, which is a guaranteed limit point~\cite{Rudin1976} (the Bolzano-Weierstrass theorem).
    

    To establish that $\overline{\uz}$ is stationary, we employ proof by contradiction. Let us assume $\overline{\uz}$ is not a stationary point. From Theorem \ref{thm:contraction}, there exist positive values $ \overline{\Upsilon} $ and $ \overline{\epsilon} $ such that
    
	%
	$$ \rho(\uz,\Upsilon) \geq \rho_0 \quad \forall \; \uz \in \cN(\overline{\uz},\overline{\epsilon}\,) \;\;\text{and}\;\; \Upsilon\in(0,\,\overline{\Upsilon}\,], $$
	where $ \rho_0 $ is chosen arbitrarily small. Without loss of generality, we can assume the entire subsequence $ \{ \uz^{k_{i}} \}$ is within $ \cN(\overline{\uz},\overline{\epsilon}) $, yielding 
	\begin{equation} \label{eq:accept}
	\rho(\uz^{k_{i}},\Upsilon) \geq \rho_0 \quad \forall\; \Upsilon \in(0,\,\overline{\Upsilon}\,].
	\end{equation}

	When the initial trust region parameter is below $ \overline{\Upsilon} $, then \eqref{eq:accept} will be automatically satisfied. However, if it exceeds $ \overline{\Upsilon} $, then multiple reductions may be necessary through the rejection mechanism  (step in line \ref{algo:step_contract} of Algorithm-\ref{alg:uav_trajectory}) before satisfying this requirement.  For each $ k_i $, we define $\widehat{\Upsilon}^{k_i}$ as the final radius that needs to be reduced with $ \widehat{\Upsilon}^{k_i} > \overline{\Upsilon} $. Additionally, let $ \Upsilon^{k_{i}} $ represent the trust region radius after the final rejection step. Now, we have  \eqref{eqn:discretized_convex_subproblem} obtained by linearizing the OCP \eqref{eqn:J_dis_Z_final} about $ \uz^{k_i} $, subject to the constraint $ \|\uz - \uz^{k_i}\|\leq \Upsilon^{k_{i}} $ where $\Upsilon^{k_{i}} = \widehat{\Upsilon}^{k_i}/\alpha >  \overline{\Upsilon}/\alpha$. 
	%
	%
	Given that $ \Upsilon $ cannot decrease beyond a minimum threshold $ \Upsilon_{\text{low}} $, we have
	\begin{equation} \label{min_delta}
	\Upsilon^{k_{i}} \geq \min \{\Upsilon_{\text{low}}, \overline{\Upsilon}/\alpha\} \triangleq \delta > 0.
	\end{equation}
	Also, the expression \eqref{eq:accept} translates to
	\begin{equation} \label{eq:r_alter}
	\widehat{J}_\lam(\uz^{k_{i}})-\widehat{J}_\lam(\uz^{k_{i}+1}) \geq \rho_{0} \Delta \cL^{k_{i}}.
	\end{equation}
	From Theorem \ref{thm:stationarity_finite} ($ \Delta \cL^{k_{i}} \geq 0 $), it is indicated that the penalized cost $\widehat{J}_\lam$ is monotonically decreasing.
	
	Next, to establish a lower bound for $ \Delta \cL^{k_{i}} $, consider $ \Delta{\uz}^\star $ as the solution of the OCP \eqref{eqn:discretized_convex_subproblem} around $ \overline{\uz} $ with trust region radius $ \delta/2 $  (satisfying  $ \|\Delta{\uz}^\star\| \leq \delta/2 $). Defining $ \widehat{\uz}\coloneqq\overline{\uz}+\Delta{\uz}^\star $, we have

	%
	\begin{equation} \label{eq:norm1}
	\|\widehat{\uz}-\overline{\uz}\| \leq \delta/2.
	\end{equation}
	Since $ \overline{\uz} $ is ssumed non-stationary, from Theorem \ref{thm:stationarity_finite} we know that 
	\begin{equation*}
	\Delta \cL_\lam^k(\Delta{\uz}^\star)=\widehat{J}_\lam(\overline{\uz}) - \cL_\lam^k(\Delta{\uz}^\star) \triangleq \theta > 0.
	\end{equation*}
	By continuity properties of $ \widehat{J}_\lam $ and $ \cL_\lam^k $, there exists an $ i_0 > 0 $ such that for all $ i \geq i_0 $
	\begin{align}
	\widehat{J}_\lam(\uz^{k_{i}}) - \cL^{k_i}(\widehat{\uz}-\uz^{k_{i}}) &> \theta/2, \text{ and}  \label{eq:r_continuity} \\
	\|\uz^{k_{i}} - \overline{\uz}\| &< \delta/2.\label{eq:norm2}
	\end{align}
	Combining our distance bounds \eqref{eq:norm1}, \eqref{eq:norm2}, and \eqref{min_delta}, for all $ i \geq i_0 $, we have
	\begin{equation} \label{eq:tr_ki}
	\|\widehat{\uz}-\uz^{k_{i}}\| \leq \|\widehat{\uz}-\overline{\uz}\| + \|\uz^{k_{i}} - \overline{\uz}\| < \delta \leq \Upsilon^{k_{i}},
	\end{equation}
	%
	
	Introducing $ \Delta\widehat{{\uz}}^{k_{i}}\coloneqq\widehat{\uz}-\uz^{k_{i}} $, we see that that $ \Delta\widehat{{\uz}}^{k_{i}} $ represents a feasible solution for the convex subproblem \eqref{eqn:discretized_convex_subproblem} at point  $ (\uz^{k_{i}},\Upsilon^{k_{i}}) $ when $ i \geq i_0 $.  If $ \Delta{\uz}^{k_{i}} $ denotes the optimal solution to this subproblem, then $ \cL^{k_i}(\Delta{\uz}^{k_{i}}) \leq \cL^{k_i}(\Delta\widehat{{\uz}}^{k_{i}}) $, yielding
	\begin{align} \label{eq:theta/2}
	\Delta \cL^{k_i}(\Delta{\uz}^{k_{i}}) &= \widehat{J}_\lam(\uz^{k_{i}}) - \cL^{k_i}(\Delta{\uz}^{k_{i}}) \notag \\
	&\geq \widehat{J}_\lam(\uz^{k_{i}}) - \cL^{k_i}(\Delta\widehat{{\uz}}^{k_{i}}) \notag \\
	&> \theta/2.
	\end{align}
	%
    With \eqref{eq:r_alter} and \eqref{eq:theta/2},  for all $ i \geq i_0 $:
	\begin{equation} \label{eq:last_theta/2}
	\widehat{J}_\lam(\uz^{k_{i}})-\widehat{J}_\lam(\uz^{k_{i}+1}) \geq \rho_{0} \theta/2.
	\end{equation}
	However, since $k_i+1 \leq k_{(i+1)}$, and the penalized cost never increases (refer \eqref{eq:r_alter}), we have $ \widehat{J}_\lam(\uz^{k_{i}+1}) \geq \widehat{J}_\lam(\uz^{k_{(i+1)}}) $, which gives us
	\begin{align*}
	\sum_{i=1}^{\infty} \left( \widehat{J}_\lam(\uz^{k_{i}})-\widehat{J}_\lam(\uz^{k_{i}+1}) \right) &\leq 
	\sum_{i=1}^{\infty} \left( \widehat{J}_\lam(\uz^{k_{i}})-\widehat{J}_\lam(\uz^{k_{(i+1)}}) \right) \\
	&=\widehat{J}_\lam(\uz^{k_{1}})-\widehat{J}_\lam(\overline{\uz}) \leq \infty.
	\end{align*}
	This establishes that the series converges, necessitating:

	$$ \widehat{J}_\lam(\uz^{k_{i}})-\widehat{J}_\lam(\uz^{k_{i}+1}) \;\rightarrow\; 0. $$
This directly contradicts our earlier inequality  $\widehat{J}_\lam(\uz^{k_{i}})-\widehat{J}_\lam(\uz^{k_{i}+1}) \geq \rho_{0} \theta/2$. 	Therefore, our contradiction proves that every limit point $ \overline{\uz} $  must be a stationary point of OCP \eqref{eqn:J_dis_Z_final}
    
\end{IEEEproof}


\section{Numerical Simulation}
\label{sec:simulation}
This section presents comprehensive numerical simulations to validate the effectiveness of the proposed methodology. The simulations employ a quadrotor UAV model with physical parameters specified in Table \ref{tab:uav_params}, while the communication channel  parameters are detailed in Table \ref{tab:comm_params}. The mission scenario requires the UAV to navigate from an initial position $\r_0 = [0,0,10]$ to a final destination $\r_f = [500,500,100]$ (coordinates in meters) while maintaining reliable data relay communications with a ground station located at $\r_{\text{GS}} = [200,400,0]$.
The numerical solution is obtained using a direct collocation method with 100 evenly distributed temporal nodes. The algorithm's convergence parameters are set as follows: maximum iterations (iter\_max) = 50, convergence tolerance $\epsilon = 10^{-4}$, and initial trust region radius $\eta = 1$. To thoroughly evaluate the algorithm's performance, two distinct simulation scenarios are investigated: first, the generation of an optimal trajectory in obstacle-free space, and second, optimal trajectory generation in an urban environment with multiple obstacles.

\begin{table}[!t]
\centering
\caption{Quadrotor UAV Parameters}
\label{tab:uav_params}
\begin{tabular}{@{}lc@{}}
\hline
\textbf{Parameter} & \textbf{Value} \\
\hline
Mass ($m$) & 3 kg \\
Arm length ($l$) & 0.3 m \\
Moment of inertia ($\mathbf{J}$) & diag(0.04, 0.04, 0.08) kg$\cdot$m$^2$ \\
Maximum rotor speed ($\omega_{\text{max}}$) & 600 rad/s \\
Maximum roll angle ($\phi_{\text{max}}$) & 35$^\circ$ \\
Maximum pitch angle ($\theta_{\text{max}}$) & 35$^\circ$ \\
\hline
\end{tabular}
\end{table}

\begin{table}[!t]
\centering
\caption{Communication Model Parameters}
\label{tab:comm_params}
\begin{tabular}{@{}lc@{}}
\hline
\textbf{Parameter} & \textbf{Value} \\
\hline
Transmission power ($P_{\text{tx}}$) & 5 W \\
Path loss exponent ($\beta$) & 2.3 \\
Ground station location ($\mathbf{r}_{\text{GS}}$) & [200, 400, 0] m \\
Atmospheric parameters ($a_1$, $a_2$) & 10, 0.6 \\
Reflection attenuation ($\zeta$) & 0.2 \\
Channel bandwidth ($w$) & 1 MHz \\
Constant parameter $(b)$ & 60  \\
\hline
\end{tabular}
\end{table}

\subsection{Optimal UAV trajectory in obstacle-free space}
In the first simulation scenario, we evaluate the UAV's trajectory optimization under different data relay requirements without any obstacles. Figure \ref{fig:free_space} shows both top and isometric views of the generated trajectories. Three different data relay scenarios were simulated where the UAV needs to transfer 30MB (solid blue line), 50MB (dashed red line), and 70MB (dash-dotted purple line) data. The resulting mission completion times were 502s, 615s, and 825s for 30MB, 50MB, and 70MB data requirements, respectively. This progressive increase in mission duration correlates directly with the increasing data transfer requirements, as evident from the increasingly deviating trajectories from the initial guess (straight line) path.

\begin{figure}[ht]
\centering
\subfigure[Top view]{\includegraphics[width=0.7\linewidth]{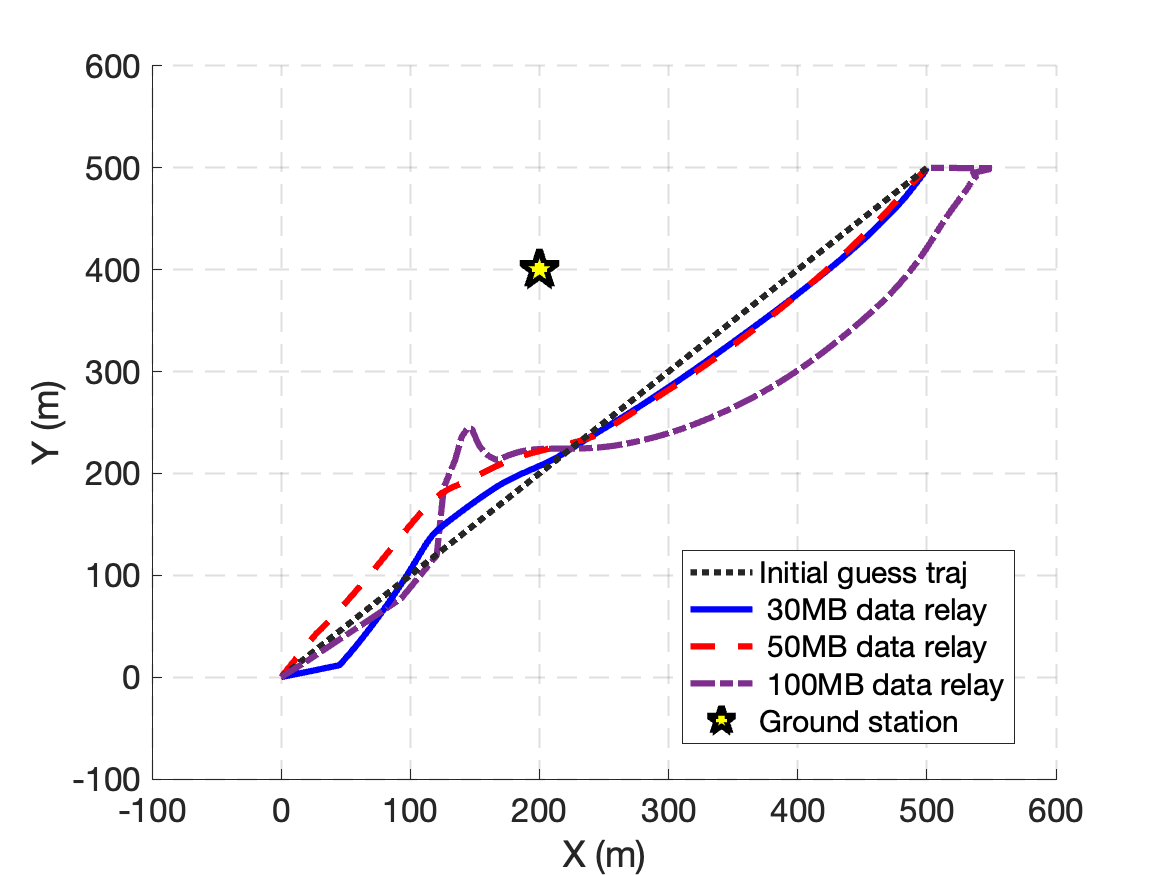}}\\
\subfigure[Isometric view]{\includegraphics[width=0.7\linewidth]{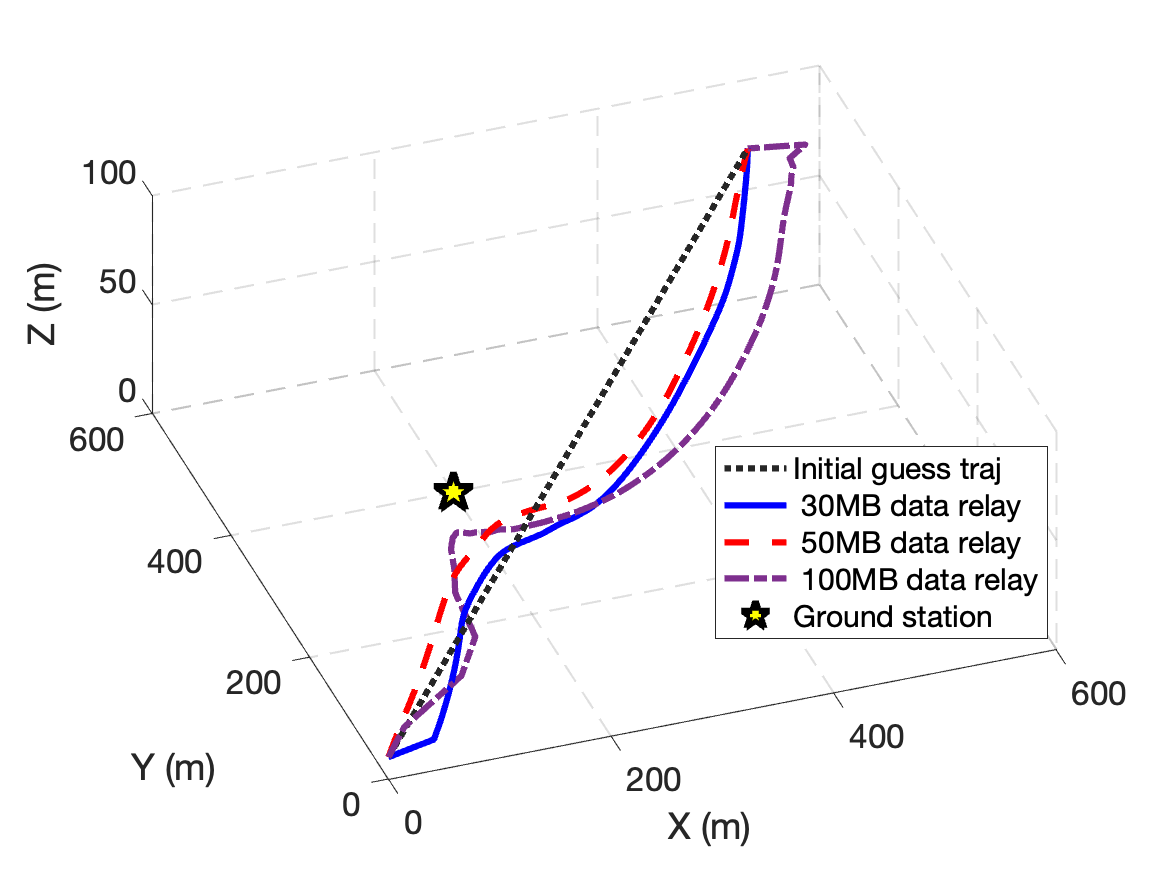}}
\caption{Optimal UAV trajectories for different data relay requirements (30MB, 50MB, 70MB) in an obstacle-free environment.}
\label{fig:free_space}
\end{figure}

The trajectories demonstrate how higher data requirements lead to more pronounced deviations from the initial guess trajectory (dotted black line). This behavior is expected as the UAV needs to maintain closer proximity to the ground station for longer durations to satisfy the increased data transfer requirements while ensuring reliable communication links.

\subsection{Optimal UAV trajectory in an urban environment}
The second simulation scenario introduces two buildings represented as obstacles in the UAV's path, positioned at coordinates [190,200,0] and [420,400,0]. The buildings are approximated by their cylindrical circumference with a safe distance constraint of $d_\text{s} = 40$m. Figure \ref{fig:urban_env} presents the optimal trajectory for a 30MB data relay requirement, where the algorithm successfully generates a path (solid blue line) that deviates from the initial guess (dotted black line) to avoid the obstacles while ensuring the required data transfer. The mission time of 722s is notably longer than the obstacle-free case (502s) for the same 30MB data requirement, which is expected as the UAV takes a more deviated and longer path to avoid the obstacles.
\begin{figure}[ht]
\centering
\subfigure[Top view]{\includegraphics[width=0.6\linewidth]{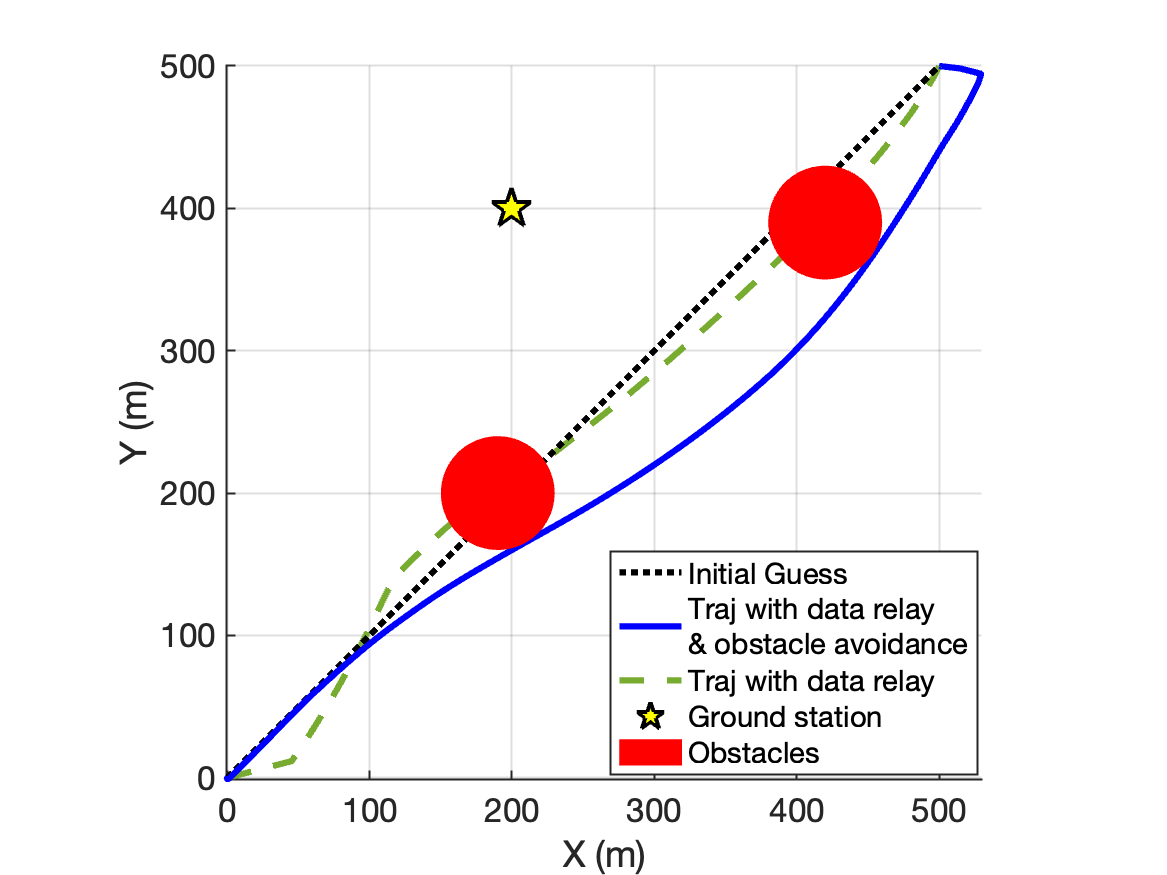}}\\
\subfigure[Isometric view]{\includegraphics[width=0.7\linewidth]{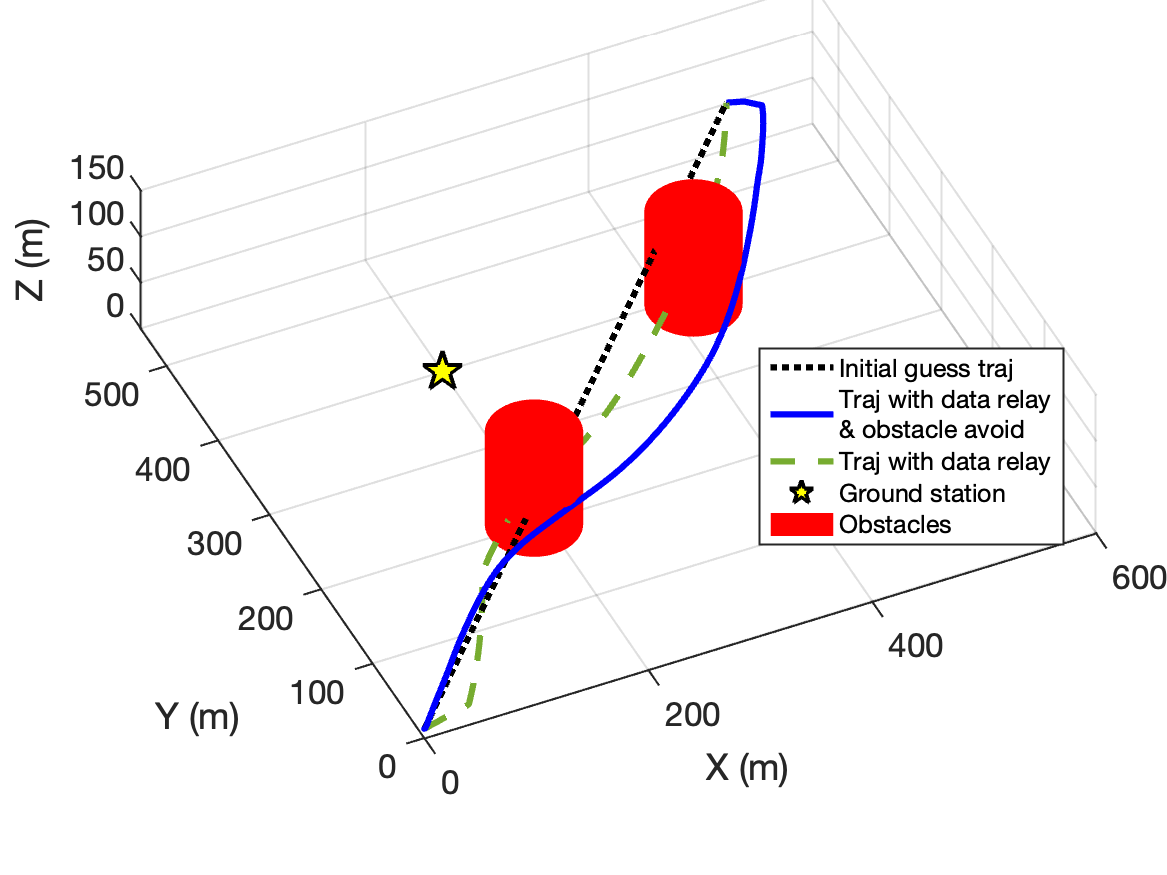}}
\caption{Optimal UAV trajectory with 30MB data relay requirement in urban environment with cylindrical obstacles.}
\label{fig:urban_env}
\end{figure}

The presence of obstacles significantly influences the trajectory generation, requiring more complex path planning to satisfy multiple competing constraints. The optimal trajectory demonstrates the algorithm's ability to balance these objectives effectively: maintaining safe distances from the buildings while ensuring the complete transfer of required data. The results validate the effectiveness of the proposed free-time OCP approach in generating optimal trajectories that satisfy multiple mission constraints in both simple and complex environments.

\section{Conclusion}
\label{sec:conclusion}


This paper presents a comprehensive framework for generating energy-efficient UAV trajectories while ensuring reliable data relay capabilities in urban environments. The methodology addresses key challenges by formulating a free-time optimal control problem that integrates full rigid-body quadrotor dynamics, realistic urban communication models, and obstacle-avoidance constraints. Sequential convex programming converts the inherently non-convex problem into tractable convex subproblems, enabling simultaneous handling of multiple mission requirements.


The framework's effectiveness is demonstrated through extensive numerical simulations in different scenarios, including obstacle-free and urban environments. Results indicate the algorithm's capacity to adapt to varying data relay requirements, with mission times scaling from 502 seconds to 825 seconds as data demands increase from 30 MB to 70 MB in obstacle-free conditions. Urban obstacles necessitate more complex trajectories, as evidenced by an increased mission time of 722 seconds for a 30 MB data requirement, highlighting the algorithm's ability to balance multiple mission objectives. The numerical tractability and ease of implementation make the framework well-suited for autonomous UAV operations in urban settings, where reliable communication and effective obstacle avoidance are critical.

\bibliographystyle{IEEEtran} 
\bibliography{reference}

\end{document}